\documentclass[pra,twocolumn,superscriptaddress,showpacs,10pt]{revtex4}

\usepackage[T1]{fontenc}
\usepackage{times}
\usepackage{color,graphicx}
\usepackage{subfig,array}
\usepackage{amsthm,amssymb,amsmath}

\newcommand\ket[1]{\ensuremath{|#1\rangle}}
\newcommand\bra[1]{\ensuremath{\langle#1|}}

\newcommand\oprod[2]{\ensuremath{|#1\rangle\langle#2|}}

\newcommand\tr{\mathop{\rm tr}\nolimits}
\newcommand\Tr{\mathop{\rm Tr}\nolimits}

\newcommand\rank{\mathop{\rm rank}\nolimits}

\def\C{\mathbb{C}}

\def\N{\mathbb{N}}

\def\R{\mathbb{R}}

\DeclareMathOperator{\im}{im}
\DeclareMathOperator{\diag}{diag}

\def\A{{\mathcal A}}

\newtheorem{lemma}{Lemma}
\newtheorem{theorem}{Theorem}
\newtheorem{corollary}{Corollary}

\newtheorem{proposition}{Proposition}

\newtheorem{example}{Example}

\begin{document}

%% End-Of-Header

\title{Uniqueness of Quantum States Compatible with Given Measurement Results}

\author{Jianxin Chen}%
\affiliation{Department of Mathematics \& Statistics, University of
  Guelph, Guelph, Ontario, Canada}%
\affiliation{Institute for Quantum Computing, University of Waterloo,
  Waterloo, Ontario, Canada}%
%\affiliation{Academy of Mathematics and Systems Science,
% Chinese Academy of Sciences, Beijing, China}%
\author{Hillary Dawkins}
\affiliation{Department of Mathematics \& Statistics, University of
  Guelph, Guelph, Ontario, Canada}%
\author{Zhengfeng Ji}%
\affiliation{Institute for Quantum Computing, University of Waterloo,
  Waterloo, Ontario, Canada}%
\affiliation{State Key Laboratory of Computer Science, Institute of
  Software, Chinese Academy of Sciences, Beijing, China}%
\author{Nathaniel Johnston}%
\affiliation{Department of Mathematics \& Statistics, University of
  Guelph, Guelph, Ontario, Canada}%
\affiliation{Institute for Quantum Computing, University of Waterloo,
  Waterloo, Ontario, Canada}%
\author{David Kribs}%
\affiliation{Department of Mathematics \& Statistics, University of
  Guelph, Guelph, Ontario, Canada}%
\affiliation{Institute for Quantum Computing, University of Waterloo,
  Waterloo, Ontario, Canada}%
\author{Frederic Shultz}
\affiliation{Mathematics Department, Wellesley College, Wellesley,
  Massachusetts, USA}%
\author{Bei Zeng}%
\affiliation{Department of Mathematics \& Statistics, University of
  Guelph, Guelph, Ontario, Canada}%
\affiliation{Institute for Quantum Computing, University of Waterloo,
  Waterloo, Ontario, Canada}%

\begin{abstract}
  We discuss the uniqueness of quantum states compatible with given
  measurement results for a set of observables. For a given pure
  state, we consider two different types of uniqueness: (1) no other
  pure state is compatible with the same measurement results and (2)
  no other state, pure or mixed, is compatible with the same
  measurement results. For case (1), it was known that for a
  $d$-dimensional Hilbert space, there exists a set of $4d-5$
  observables that uniquely determines any pure state. We show that
  for case (2), $5d-7$ observables suffice to uniquely determine any
  pure state. Thus there is a gap between the results for (1) and (2),
  and we give some examples to illustrate this. Unique determination
  of a pure state by its reduced density matrices (RDMs), a special
  case of determination by observables, is also discussed. We improve
  the best known bound on local dimensions in which almost all pure
  states are uniquely determined by their RDMs for case (2). We
  further discuss circumstances where (1) can imply (2). We use
  convexity of the numerical range of operators to show that when only
  two observables are measured, (1) always implies (2). More
  generally, if there is a compact group of symmetries of the state
  space which has the span of the observables measured as the set of
  fixed points, then (1) implies (2). We analyze the possible
  dimensions for the span of such observables. Our results extend
  naturally to the case of low rank quantum states.
\end{abstract}

\date{\today}


\maketitle

\section{I. Introduction}

In a $d$-dimensional Hilbert space $\mathcal{H}_d$, the description of
any quantum state $\rho$ generated by a source can be obtained by
quantum tomography. For any density matrix $\rho$, which is Hermitian
and has trace $1$, $d^2-1$ independent measurements are sufficient and
necessary to uniquely specify $\rho$. When $\rho=\ket{\psi}\bra{\psi}$
is a pure state, one may not need as many measurements to uniquely
determine $\ket{\psi}$. As we will see later, however, exactly what is
meant by ``uniquely'' in this context needs to be specified.

Consider a set of $m$ linearly independent observables
\begin{equation}
\label{eq:A}
\mathbf{A}=(A_1,A_2,\ldots,A_m)
\end{equation}
where each $A_i$ is Hermitian. Measurements on state $\rho$ with
respect to these observables give the following average values
\begin{equation}
\mathbf{A}(\rho):=(\tr\rho A_1,\tr\rho A_2,\ldots,\tr\rho A_m) \in
\mathbb{R}^m.
\end{equation}
We denote the set of these $\mathbf{A}(\rho)$ for all states $\rho$ as
\begin{equation}
\label{eq:Cm}
C_m(\mathbf{A}):=\{\mathbf{A}(\rho):\, \rho \text{ acts on }
\mathcal{H}_d\}.
\end{equation}
For a pure state $\ket{\psi}$, these values are given by
\begin{equation}
\mathbf{A}(\ket{\psi}):=(\bra{\psi} A_1\ket{\psi},\bra{\psi}
A_2\ket{\psi},\ldots, \bra{\psi} A_m\ket{\psi}),
\end{equation}
and we denote the set of these values for all pure states $\ket{\psi}$
as the joint numerical range
\begin{equation}
\label{eq:Wm}
W_m(\mathbf{A}):=\{\mathbf{A}(\ket{\psi}):\, \ket{\psi}\in\mathcal{H}_d\}.
\end{equation}

In this work we consider two different kinds of ``unique
determinedness'' for $\ket{\psi}$:
\begin{enumerate}
\item We say $\ket{\psi}$ is \emph{uniquely determined among pure
    states (UDP) by measuring $\mathbf{A}$} if there does not exist
  any other pure state which has the same measurement results as those
  of $\ket{\psi}$ when measuring $\mathbf{A}$.
\item We say $\ket{\psi}$ is \emph{uniquely determined among all
    states (UDA) by measuring $\mathbf{A}$} if there does not exist
  any other state, pure or mixed, which has the same measurement
  results as those of $\ket{\psi}$ when measuring $\mathbf{A}$.
\end{enumerate}

It is known that there exists a family of $4d-5$ observables such that
any pure state is UDP, in contrast to the $d^2-1$ observables in the
general case of quantum tomography ~\cite{HTW11}. The physical meaning
for this case is clear: it is useful for the purpose of quantum
tomography to have the prior knowledge that the state to be
reconstructed is pure or nearly pure. Many other techniques for pure
state tomography have been developed, and experiments have been
performed to demonstrate the reduction of the number of measurements
needed~\cite{Wei92,AW99,Fin04,FSC05,GLF+10,CPF+10,LZL+12}.

When the state is UDP, to make the tomography meaningful, one needs to
make sure that the state is indeed pure. This is not in general
practical, but one can readily generalize the above mentioned UDP
results to low rank states, where the physical constraints (e.g., low
temperature, locality of interaction) may ensure that the actual
physical state (which ideally supposed to be pure) is indeed low rank.
If the state is UDA, however, in terms of tomography one do not need
to bother with these physical assumptions, because in the event there
is only a unique state compatible with the measurement results, which
turns out to be pure (or low rank).

There is also another clear physical meaning for the states that are
UDA by measuring $\mathbf{A}$. Consider a Hamiltonian of the form
\begin{equation}
H_{\mathbf{A}}=\sum_{i=1}^{m} \alpha_i A_i.
\end{equation}
Then any unique ground state $\ket{\psi}$ of $H_{\mathbf{A}}$ is UDA
by measuring $\mathbf{A}$. This is easy to verify: if there is any
other state $\rho$ that gives the same measurement results, then
$\rho$ has the same energy as that of $\ket{\psi}$, which is the
ground state energy. Therefore, any pure state in the range of $\rho$
must also be a ground state, which contradicts the fact that
$\ket{\psi}$ is the unique ground state. In other words, UDA is a
necessary condition for $\ket{\psi}$ to be a unique ground state of
$H_{\mathbf{A}}$. It is in general not sufficient, but the exceptions
are likely rare~\cite{CJZ+11,CJR+12}.

The uniqueness properties for pure states, for both UDP and UDA, have
also been studied extensively in the case of multipartite quantum
systems, where the observables correspond to reduced density matrices
(RDMs). That is, the observables are chosen to act nontrivially on
only some subsystems. For an $n$-particle system and a constant $k<n$,
there are a total of ${n\choose k}$ $k$-RDMs, and the corresponding
measurements $\mathbf{A}$ are those $\leq k$-body operators. For
example, for a three-qubit system and $k=2$, one can choose
$\mathbf{A}$ as all the one and two-particle Pauli operators. Of
course, one can also choose to look at some of the $n\choose k$-RDMs,
rather than all of them. For instance, for a three-particle system,
one can look at $2$-RDMs of particle pairs $\{1,2\}$ and $\{1,3\}$.

It is known that almost all three-qubit pure states are UDA by their
$2$-RDMs~\cite{LPW02}. These authors also show that UDP implies UDA
for three-qubit pure states, for $2$-RDMs. This result can be further
improved to $2$-RDMs of particle pairs $\{1,2\}$ and
$\{1,3\}$~\cite{CJZ+11}. More generally one can consider a
three-particle system of particles $1,2,3$ with Hilbert spaces whose
dimensions are $d_1,d_2, d_3$, respectively. If $d_1\geq d_2+d_3-1$,
then almost all pure states are UDA by their $2$-RDMs of particle
pairs $\{1,2\}$ and $\{1,3\}$. In contrast, if $d_1\geq 2$, then
almost all pure states are UDP by their $2$-RDMs of particle pairs
$\{1,2\}$ and $\{1,3\}$, as shown by Diosi \cite{Dio04}.

For $n$-particle quantum systems with equal dimensional subsystems,
almost all pure states are UDA by their $k$-RDMs of just over half of
the parties (i.e., $k\sim n/2$). Furthermore, $\sim n/2$ properly
chosen RDMs among all the $n\choose k$ $k$-RDMs suffice~\cite{JL05}.
W-type states are UDA by their $2$-RDMs, and $n-1$ of those $2$-RDMs
are enough~\cite{PR08}. General symmetric Dicke states are UDA by
their $2$-RDMs~\cite{CJW+12}. It has been shown that the only
$n$-particle pure states which cannot be UDP by their $(n-1)$-RDMs are
those GHZ-type states, and the result is further improved to the case
of UDA~\cite{WL08}. Their results also show that UDP implies UDA for
$n$-qubit pure states, for $(n-1)$-RDMs.

Despite these many results, there is no systematic study of these two
different types of uniqueness for pure states. This will be the focus
of this paper, where we are interested in knowing for given
measurements $\mathbf{A}$, whether UDP and UDA are the same, or are
different. We will give a general argument that there is a gap between
the number of observables needed for the two different cases. However,
in many interesting circumstances, they can coincide. Our discussions
extend naturally to the case of low rank quantum states instead of
just pure states. Here one can also look at two kinds of uniqueness
when measuring given observables $\mathbf{A}$: one is uniqueness among
all low rank states, the other is among all states of any rank.

We organize the paper as follows. In Sec. II, we first show that there
is a set of $5d-7$ observables that insures every pure state is UDA;
which should be compared to the UDP result $4d-5$. Thus in general
there is a gap between the optimal results for the UDP and UDA cases,
and we illustrate this with some examples. Sec. III discusses the case
of observables corresponding to RDMs of a multipartite quantum state,
where for the three particle case, we show that if
$d_1\geq\min(d_2,d_3)$, then almost all pure states are UDA by their
$2$-RDMs of particle pairs $\{1,2\}$ and $\{1,3\}$, improving the
bounds given in \cite{LW02}. However this still leaves a gap with the
Diosi result for the case of UDP in \cite{Dio04}. We further discuss
circumstances where UDP can imply UDA for all pure states. In Sec. IV,
we show that when there are only two independent measurements
performed, then UDP always implies UDA, by making use of convexity of
the numerical range of operators. In a more general case, if there is
a compact group of symmetries of the state space which has the span of
the operators measured as its set of fixed points, then UDP implies
UDA for all pure states. We analyze the possible dimensions for those
fixed point sets. A summary and some discussions are included in Sec.
VI.

\section{II. The number of observables for UDA}

In this section, we discuss the minimum number of observables needed
to have all pure states be UDA. We start by choosing a Hermitian basis
$\{\lambda_i\}_{i=0}^{d^2-1}$ for the operators on $\mathcal{H}_d$.
Without loss of generality we choose $\lambda_0=\sqrt{d-1}I$, the
identity operator on $\mathcal{H}_d$, which has trace $d$. We further
require that the $\lambda_i$'s are orthogonal, in the sense that for
$i,j\geq 0$,
\begin{equation}
\tr\lambda_i\lambda_j=d(d-1)\delta_{ij}.
\end{equation}

The $d \times d$ Hermitian matrices form a real inner product space
with inner product $\langle A, B\rangle= \tr(AB)$, so such a basis
$\{\lambda_i\}_{i=0}^{d^2-1}$ exists for any dimension $d$. For
instance, for the qubit case ($d=2$), we can choose the Pauli basis
\begin{equation}
\label{eq:Pauli}
\lambda_1=\begin{pmatrix}0 & 1\\1& 0\end{pmatrix},\
\lambda_2=\begin{pmatrix}0 & -i\\i& 0\end{pmatrix},\
\lambda_3=\begin{pmatrix} 1 & 0\\ 0& -1\end{pmatrix}.
\end{equation}

For the qutrit case ($d=3$), one can choose $\lambda_i=\sqrt{3}M_i$
for $i>0$, where $M_i$s are the Gell-Mann matrices given by

\begin{eqnarray}
\label{eq:GellMan}
\begin{array}{ccc}
M_1	=
\left(
\begin{array}{ccc}
0  & 1  & 0  \\
1  & 0  & 0  \\
0  & 0  & 0
\end{array}
\right),
&
M_2	=
\left(
\begin{array}{ccc}
0  & -i  & 0  \\
i  & 0  & 0  \\
0  & 0  & 0
\end{array}
\right),
\\
M_3	=
\left(
\begin{array}{ccc}
1  & 0  & 0  \\
0  & -1  & 0  \\
0  & 0  & 0
\end{array}
\right),
&
M_4	=
\left(
\begin{array}{ccc}
0  & 0  & 1  \\
0  & 0  & 0  \\
1  & 0  & 0
\end{array}
\right),
\\
M_5	=
\left(
\begin{array}{ccc}
0  & 0  & -i  \\
0  & 0  & 0  \\
i  & 0  & 0
\end{array}
\right),
&
M_6	=
\left(
\begin{array}{ccc}
0  & 0  & 0  \\
0  & 0  & 1  \\
0  & 1  & 0
\end{array}
\right),\\
M_7	=
\left(
\begin{array}{ccc}
0  & 0  & 0  \\
0  & 0  & -i  \\
0  & i  & 0
\end{array}
\right),
& M_8	=
\frac{1}{\sqrt{3}}\left(
\begin{array}{ccc}
1  & 0  & 0  \\
0  & 1  & 0  \\
0  & 0  & -2
\end{array}
\right).
  &
\end{array}
\end{eqnarray}

For general $d$, one can choose $\lambda_i=\sqrt{\frac{d(d-1)}{2}}M_i$
for $i>0$, where $M_i$s are the generalized Gell-Man matrices.

We can now write any density operator $\rho$ as
\begin{equation}
\label{eq:rho}
\rho=\frac{1}{d}(I+\vec{r}\cdot\vec{\lambda}),
\end{equation}
where $\vec{\lambda}=(\lambda_1,\lambda_2,\ldots,\lambda_{d^2-1})$,
and where $\vec{r}=(r_1,r_2,\ldots,r_{d^2-1})$ has real entries.

We have $\tr\rho^2\leq 1$, therefore $\vec{r}\cdot\vec{r}\leq 1$, and
the equality holds if $\rho$ is a pure state. However, not every state
satisfying $\vec{r}\cdot\vec{r}=1$ is a pure state. Indeed, $\rho$ is
a pure state if and only if $\rho^2=\rho$, which gives equations that
$\vec{r}$ needs to satisfy.

If one of the observables is a multiple of the identity, then we can
drop it from the list of observables without affecting UDA and UDP. If
two states agree on an observable $A_i$, then they agree on $A_i+ t
I$ for any real scalar $t$, so we can adjust each of the observables
$\mathbf{A} = (A_1, \ldots, A_m)$ to have trace zero without affecting
UDA or UDP. Hence hereafter we assume all $A_i$ are traceless.

For any observable $A_i$, we can expand in terms of $\{\lambda_i\}$ as
\begin{equation}
A_i=\sum_{j=1}^{d^2-1} \alpha_{ij}\lambda_j.
\end{equation}

Then the average value of $A_i$ is given by
\begin{equation}
  \tr(A_i\rho) = \frac{1}{d}(d+\sum_j r_j\alpha_{ij}d(d-1)) =
  1+(d-1)\vec{r}\cdot\vec{\alpha_i},
\end{equation}
where $\vec{\alpha_i} = \{\alpha_{i1}, \alpha_{i2}, \ldots,
\alpha_{i(d^2-1)}\}$.

To discuss the problem for any pure state to be UDA, the constant $1$
and constant factor $d-1$ can be ignored, as these are the same
constants for all states. Therefore we have
\begin{equation}
\tr(A_i\rho)\sim \vec{r}\cdot\vec{\alpha_i},
\end{equation}
where $\sim$ means that the average value of $A_i$ for the state
$\rho$ is geometrically equivalent to the projection of $\vec{r}$ onto
$\vec{\alpha_i}$.

Alternatively, define $T:\mathbb{R}^{d^2-1} \to \mathbb{R}^m$ by
$T(\vec r) = (\vec r\cdot \alpha_1, \ldots, \vec r \cdot \vec
\alpha_m)$. Let $L$ be the linear subspace of $\mathbb{R}^{d^2-1}$
spanned by $\vec \alpha_1, \ldots, \vec \alpha_m$, and let $\pi$ be
the orthogonal projection from $\mathbb{R}^{d^2-1}$ onto $L$. Then
$\pi$ and $T$ have the same kernel, namely $L^\perp$. Thus for states
$\rho_1, \rho_2$, we have $T(\rho_1) = T(\rho_2)$ if and only if
$\pi(\rho_1) = \pi(\rho_2)$, so in considering UDA and UDP we can
treat $T$ as being the orthogonal projection onto $L$.

If we subtract the density matrix $I/d$ from all states, then the
translated set of states sits in the real $d^2-1$ dimensional subspace
of trace zero Hermitian matrices. In this sense, we are actually
working with real geometry in $\mathbb{R}^{d^2-1}$. All quantum states
then sit inside the $d^2-1$-dimensional unit ball, with pure states
corresponding to unit vectors, but not every vector on the unit
$d^2-2$-dimensional sphere is a pure state. The observables span an
$m$-dimensional subspace that all the quantum states will be projected
onto. We will simply say the subspace is spanned by $\mathbf{A}$ when
no confusion arises, and we will no longer distinguish an operator
$A_i$ from the corresponding vector $\vec{\alpha_{i}}$. Indeed we only
consider the real span of $\mathbf{A}$, and we denote it by
$\mathcal{S}(\mathbf{A})$. For each $\mathcal{S}(\mathbf{A})$, there
is an orthogonal subspace in $\mathbb{R}^{d^2-1}$ of dimension
$d^2-1-m$, which we denote by $\mathcal{S}(\mathbf{A})^{\perp}$. Here
we are taking the orthogonal complement in the space of traceless
Hermitian matrices, so that every
$V\in\mathcal{S}(\mathbf{A})^{\perp}$ is traceless.

We now are ready to state our first theorem.

\begin{theorem}
\label{th:UDA}
For a $d$-dimensional system ($d>2$), there exists a set of $5d-7$
observables for which every pure state is UDA.
\end{theorem}

To see why this is the case, note that in the above-mentioned
geometrical picture, it is clear that a pure state
$\ket{\psi}\bra{\psi}$ is UDA by measuring $\mathbf{A}$ if there does
not exist any operator $V\in\mathcal{S}(\mathbf{A})^{\perp}$, such
that $\ket{\psi}\bra{\psi}+V$ is positive. One sufficient condition
will then be that any operator $V\in\mathcal{S}(\mathbf{A})^{\perp}$
has at least two positive and two negative eigenvalues. We will use
this sufficient condition to construct a desired
$\mathcal{S}(\mathbf{A})^{\perp}$.

In order to construct $\mathcal{S}(\mathbf{A})^{\perp}$, we provide a
set of $m = d^2-5d+6$ linearly independent Hermitian matrices $H_1,
H_2, \ldots, H_m \in M_d(\C)$ explicitly, such that the Hermitian
matrix
\begin{equation*}
  \sum_{j=1}^m r_j H_j
\end{equation*}
has at least two positive eigenvalues for any nonzero real vector
$r=(r_j) \in \R^m$.

Our construction is motivated by and similar to the diagonal filling
technique used in Ref.~\cite{CMW08}, but along the other direction of
the diagonals.

This then means that measuring $d^2-1-(d^2-5d+6)=5d-7$ observables is
enough for any pure state to be UDA, which proves the theorem. There
are indeed technical details to be clarified that we leave to Appendix
A.

If we compare our results with those given in ~\cite{HTW11}, which
shows that measuring $4d-5$ observables are enough for any pure state
to be UDP, there exists an obvious gap. We claim that this gap indeed
cannot be closed in general. To see this, let us look at the simplest
case of $d=3$, where the results just compared state that $7$
observables are enough for any pure state to be UDP but $8$
observables are enough for any pure state to be UDA.

If one can measure a particular set $\mathbf{A}$ with $7$ observables
and have all pure states be UDA, then also every state also must be
UDP for measuring $\mathbf{A}$. According to~\cite{HTW11}, this only
happens if $\mathcal{S}(\mathbf{A})^{\perp}$ contains a single
invertible traceless operator $V$, meaning $V$ is rank $3$. Without
loss of generality we can assume the largest eigenvalue $V$ to be
positive with an eigenstate $\ket{\psi}$. Then $\ket{\psi}$ is not UDA
by measuring $\mathbf{A}$ since as observed in ~\cite{HTW11} there
exists a mixed state which also has the same average values as those
of $\ket{\psi}$. Therefore, one cannot only measure $7$ observables
for all pure states to be UDA.

For general $d$, our construction needs $5d-7$ observables. We do not
know whether this is the optimal construction, but it is very unlikely
one can get this down to $4d-5$. In other words, in general UDA and
UDP for pure states should be indeed two different concepts and there
should always be gaps between the number of observables needed to be
measured for each case to uniquely determine any pure quantum state.
This is one exception though, which is for the qubit case (i.e.,
$d=2$) where it is shown in~\cite{HTW11} that for all pure states to
be UDP, one needs to measure $3=2^2-1$ variables, which then uniquely
determine any quantum state among all states.

Finally, we remark that our results in Theorem 1 naturally extend to
the case of low rank states. That is, for a rank $q<d/2$ quantum state
$\rho$, we can similarly consider two different cases: (1) $\rho$ is
uniquely determined by measuring $\mathbf{A}$ among all rank $\leq q$
states (which was considered in ~\cite{HTW11}) (2) $\rho$ is uniquely
determined by measuring $\mathbf{A}$ among all quantum states of any
rank.

\begin{theorem}
\label{th:UDAq}
For a $d$-dimensional system ($d>2$) measuring $(4q+1)d-(4q^2+2q+1)$
observables is enough for a rank $\leq q$ state to be uniquely
determined among all states.
\end{theorem}

Compared to the results in~\cite{HTW11}, where $4q(d-q)-1$ observables
are needed to uniquely determine any rank $\leq q$ states among all
rank $\leq q$ states, when $d$ is large the difference in the leading
term has a $d$ gap. The proof idea is similar to that of
Theorem~\ref{th:UDA}, so we leave the details to Appendix A.

\section{III. The case of reduced density matrices}

In this section we discuss the case where the Hilbert space
$\mathcal{H}_d$ is a multipartite quantum system, where the
observables correspond to the reduced density matrices (RDMs). That
is, the observables are chosen to be acting nontrivially only on some
subsystems. For instance, for a three-qubit system, the observables
corresponding to the $2$-RDMs of particle pairs $\{1,2\}$ can be
chosen as
\begin{eqnarray}
\mathbf{A}=&(X_1,X_2, Y_1,Y_2,Z_1,Z_2 \nonumber\\
&X_1X_2,X_1Y_2,X_1Z_2,Y_1X_2,Y_1Y_2,\nonumber\\
&Y_1Z_2,Z_1X_2,Z_1Y_2,Z_1Z_2),
\end{eqnarray}
where $X_i,Y_i,Z_i$ are Pauli $X,Y,Z$ operators acting on the $i$th
qubit.

For simplicity in this section we consider only $3$-particle systems,
labeled by $1,2,3$, and each with Hilbert space dimension
$d_1,d_2,d_3$, respectively. That is, $\mathcal{H}_d =
\mathcal{H}_{d_1} \otimes \mathcal{H}_{d_2} \otimes \mathcal{H}_{d_3}$
and $d=d_1d_2d_3$. Nevertheless, our method naturally extends to
systems of more than $3$-particles.

Recall that for a three particle system, it is known that almost all
three-qubit pure states are UDA by their $2$-RDMs~\cite{LPW02}. This
result can be further improved to $2$-RDMs of particle pairs $\{1,2\}$
and $\{1,3\}$~\cite{CJZ+11}. More generally, if $d_1\geq d_2+d_3-1$,
then almost all pure states are UDA by their $2$-RDMs of particle
pairs $\{1,2\}$ and $\{1,3\}$~\cite{LW02}. In contrast, if $d_1\geq
2$, then almost every pure state is UDP by its $2$-RDMs of particle
pairs $\{1,2\}$ and $\{1,3\}$~\cite{Dio04}.

We notice that different from the discussion in Sec. II, one no longer
considers uniqueness for all pure states, but `almost all' of them.
This means there exists a measure zero set of pure states which are
not uniquely determined. For instance, for the three qubit case, any
state which is local unitarily equivalent to the GHZ type state
\begin{equation}
\ket{GHZ}_{\text{type}}=a\ket{000}+b\ket{111}
\end{equation}
cannot be UDP, as any state of the form
$a\ket{000}+be^{i\theta}\ket{111}$ has the same $2$-RDMs as those of
$\ket{GHZ}_{\text{type}}$. This means that, for a three qubit pure state
$\ket{\psi}$, it is either UDA, or not UDP. In other words, if any
three qubit pure state $\ket{\psi}$ is UDP, then it is UDA by its
$2$-RDMs of particle pairs $\{1,2\}$ and $\{1,3\}$. In this sense, we
say in this case UDP implies UDA for all pure states.

However, for the general case of a three particle system, there is a
gap between known results of UDA and UDP. Our following result
improves the bound for the UDA case.

\begin{theorem}
\label{th:RDM1}
If $d_1\geq\min (d_2,d_3)$, then almost every tripartite quantum state
$\ket{\phi}\in
\mathcal{H}_{d_1}\otimes\mathcal{H}_{d_2}\otimes\mathcal{H}_{d_3}$ is
UDA by its $2$-RDMs of particle pairs $\{1,2\}$ and $\{1,3\}$,
\end{theorem}

To see why this is the case, an arbitrary pure state $\ket{\phi}$ of
this system can be written as
\begin{eqnarray}
\label{eq:phi123}
\ket{\phi}_{123} = \sum\limits_{i=1}^{d_{1}} \sum\limits_{j=1}^{d_{2}}
\sum\limits_{k=1}^{d_{3}} c_{ijk}\ket{i}_{1}\ket{j}_{2}\ket{k}_{3}.
\end{eqnarray}
If there is another state $\rho$ which agrees with $\ket{\phi}$ in its
subsystems $\{1,2\}$ and $\{1,3\}$, then we can find a pure state
$\ket{\psi}_{1234}\in\mathcal{H}_{d_1}\otimes \mathcal{H}_{d_2}\otimes
\mathcal{H}_{d_3}\otimes \mathcal{H}_{d_4}$ which agrees with $\rho$
on the subsystem $\{1,2,3\}$ and also agrees with $\ket{\phi}_{123}$
in subsystems $\{1,2\}$ and $\{1,3\}$.

Since the rank of the $2$-RDM of the subsystem $\{1,2\}$ is at most
$d_{3}$, the pure state $\ket{\psi}_{1234}$ can be written as a
superposition of $\ket{v_{l}}\ket{E_{l}}$ as follows.
\begin{eqnarray}
\label{eq:psi1234}
\ket{\psi}_{1234}=\sum\limits_{l=1}^{d_{3}} \ket{v_{l}}\ket{E_{l}}
\end{eqnarray}
where
\begin{eqnarray}
\ket{v_{l}}=\sum\limits_{i=1}^{d_{1}}\sum\limits_{j=1}^{d_{2}}c_{ijl}\ket{i}_1\ket{j}_2
\end{eqnarray}
for any $1\leq l\leq d_{3}$. Here $\{\ket{E_l}\}_{i=1}^{d_{3}}$ will
be vectors (perhaps unnormalized) in $\mathcal{H}_{d_3}\otimes
\mathcal{H}_{d_4}$.

The states $\{\ket{E_l}\}_{i=1}^{d_{3}}$ can be chosen to be
orthonormal vectors in the subsystem $\mathcal{H}_{d_3}\otimes
\mathcal{H}_{d_4}$, and then for almost all states $\ket{\phi}$, the
set of $\{\ket{v_l}\}_{i=1}^{d_{3}}$ will be linearly independent. Let
us write $\ket{E_{l}} = \sum\limits_{k=1}^{d_{3}} \ket{k}_{3}
\ket{e_{lk}}_{4}$. For any $1\leq l\leq d_{3}$, we will have
\begin{eqnarray}
\ket{\psi}_{1234} = \sum\limits_{i=1}^{d_{1}}
\sum\limits_{j=1}^{d_{2}} \sum\limits_{k,l=1}^{d_{3}}
c_{ijl}\ket{i}_{1} \ket{j}_2\ket{k}_{3}\ket{e_{lk}}_{4}.
\end{eqnarray}

Now let's consider the subsystem $\{1,3\}$. Since $\ket{\phi}_{123}$
and $\ket{\psi}_{1234}$ have the same RDMs for particles $\{1,3\}$,
this gives
\begin{eqnarray}
\label{eq:RDM12}
\tr_{2}\ket{\phi}\bra{\phi}=\tr_{\{2,4\}}\ket{\psi}\bra{\psi}.
\end{eqnarray}

Substituting Eqs.~\eqref{eq:phi123} and~\eqref{eq:psi1234} into
Eq.~\eqref{eq:RDM12}, and comparing each matrix element, results in
the following equalities (for all $m, m', n, n'$):
\begin{equation}
\label{eq:eijkl}
\sum_{j=1}^{d_2}c_{mjn}c^{*}_{m'jn'} = \sum_{j=1}^{d_2}
\sum_{k,k'=1}^{d_3}c_{mjk}c^{*}_{m'jk'} \langle e_{k'n'}\ket{e_{kn}}.
\end{equation}

Now let us define $x_{ijkl}=\langle e_{ij}\ket{e_{kl}}$. Then
Eq.~\eqref{eq:eijkl} is a linear equation system with variables
$x_{ijkl}$. It is not hard to verify that
\begin{equation}
\label{eq:solution}
x_{ijkl} = \left\{ \begin{array}{ll}
1 & \textrm{if $i=j, k=l$}\\
0 & \textrm{otherwise}
\end{array} \right.
\end{equation}
is a solution to the equation system, which corresponds to the state
$\ket{\phi}_{123}$.

Now we need to show that when $d_1\geq\min\{d_2,d_3\}$,
Eq.~\eqref{eq:eijkl} has only one solution which is given by
Eq.~\eqref{eq:solution}. It turns out that this is indeed the case
which then proves Theorem~\ref{th:RDM1}. In fact, the linear equations
above are generically linearly independent. To see this, let's fix $n,
n^{\prime}$ and $m, m^{\prime}$, the right-hand side of
Eq.~\eqref{eq:eijkl} is $\sum\limits_{k,k'=1}^{d_3}
\bra{\alpha}_{m'k'} \cdot \ket{\alpha}_{mk} x_{k'n'kn}$ where
$\ket{\alpha}_{mk}=\sum\limits_{j=1}^{d_2} c_{mjk}\ket{j}$. Then the
coefficient matrix can be written as the following:
\begin{eqnarray}
\left(
\begin{array}{cccc}
\bra{\alpha}_{11}\ket{\alpha}_{11}  &  \bra{\alpha}_{11}\ket{\alpha}_{12} & \cdots  &\bra{\alpha}_{1d_3}\ket{\alpha}_{1d_3}   \\
\bra{\alpha}_{11}\ket{\alpha}_{21}  &  \bra{\alpha}_{11}\ket{\alpha}_{22} & \cdots  &\bra{\alpha}_{1d_3}\ket{\alpha}_{2d_3}   \\
\vdots  & \vdots  &  \ddots & \vdots \\
\bra{\alpha}_{d_11}\ket{\alpha}_{d_11}  &  \bra{\alpha}_{d_11}\ket{\alpha}_{d_12} & \cdots  &\bra{\alpha}_{d_1d_3}\ket{\alpha}_{d_1d_3}   \\
\end{array}
\right).
\end{eqnarray}
The $(d_1(i-1)+j,d_3(p-1)+q)$ entry in the above matrix is
$\bra{\alpha}_{ip}\ket{\alpha}_{jq}$.

If there are more than $1$ solutions, then the determinant of the
above matrix should be zero. Note that the determinant can be written
as a polynomial of $c_{mjk}$'s and $c_{m'jk'}^{\ast}$'s. Since $\prod
c_{iii}$ appears only once in the polynomial, the determinant of the
top $d_3^2$ by $d_3^2$ submatrix must be non-zero generically.
Therefore, $d_1^2d_3^2$ linear equations are sufficient to determine
$d_3^4$ variables.

However, we do not know whether the sufficient condition given by
Theorem~\ref{th:RDM1} for almost all three-particle pure state to be
UDA by its $2$-RDMs of particle pairs $\{1,2\}$ and $\{1,3\}$ is also
necessary. This still leaves a gap between the result of
Theorem~\ref{th:RDM1} for UDA, and the result for UDP in~\cite{Dio04}.
They both only coincide when $d_1=d_2=d_3=2$, i.e., the three qubit
case. It remains open for other cases, whether UDP can imply UDA.

Following a similar discussion as in Sec. II, our result in this
section also extends to uniqueness of low rank quantum states. In
particular, we have the following theorem.

\begin{theorem}
\label{th:RDM2}
Almost every tripartite density operator $\rho$ acting on the Hilbert
space $\mathcal{H}_{d_1} \otimes \mathcal{H}_{d_2} \otimes
\mathcal{H}_{d_3}$ with rank no more than $\lfloor
\frac{d_{1}}{d_{3}}\rfloor$ can be uniquely determined among all
states by its $2$-RDMs of particle pairs $\{1,2\}$ and $\{1,3\}$.
\end{theorem}

This result is to our knowledge, the first one for uniqueness of mixed
states with respect to RDMs. The proof is a direct extension of that
for Theorem~\ref{th:RDM1}, but with more lengthy details that we
will include in Appendix B.

Let us look at some consequences of Theorem~\ref{th:RDM2}. Consider a
four qubit system with qubits $1,2,3,4$, and look at the qubits $3,4$
as a single systems $3'$. Then Theorem~\ref{th:RDM2} says also that
almost all four qubit states of rank $2$ are UDA by their RDMs of
particles $\{1,2\}$ and $\{1,3'\}=\{1,3,4\}$, or one can say that
almost all four qubit states of rank $2$ are UDA by their $3$-RDMs.
This is indeed consistent with the multipartite result in~\cite{JL05}
which states that almost all four-qubit pure states are UDA by their
$3$-RDMs, and our result is indeed stronger. This demonstrates
that our analysis naturally extends to systems of more than
$3$-particles. We also remark that the rank of a state $\rho$ which
could be UDA by its $k$-RDMs needs to be relatively low, otherwise one
can always find another state $\rho'$ with lower rank which has the
same $k$-RDMs as those of $\rho$~\cite{JJK+12}.

\section{IV. The case of only two observables}

In Sec. II and Sec. III, we discussed the difference and coincidence
between the two kinds of uniqueness for pure states, UDA and UDP,
which in general are not the same thing. However, in certain
interesting circumstances such as the three qubit case with respect to
$2$-RDMs, and in general the $n$-qubit case respect to $(n-1)$-RDMs,
they do coincide. Starting from this section we would like to build
some general understanding of the circumstances when UDP implies UDA
for all pure states.

We start from the simplest case of $m=2$, where only two observables
are measured, i.e., $\mathbf{A}=(A_1,A_2)$. Intuitively, in this
extreme case almost no pure state can be uniquely determined, either
UDA or even UDP. However there are also exceptions. For instance, if
one of the observables, say $A_1$, has a nondegenerate ground state
$\ket{\psi}$, then $\ket{\psi}$ is UDA (hence, of course, UDP) even by
measuring $A_1$ only. One would hope this is the only exception, that
is, for a pure state $\ket{\psi}$, either it is UDA, or it is not UDP,
when only two observables are measured. We make this intuition
rigorous by the following theorem.

\begin{theorem}
\label{th:Num}
When only two observables are measured, i.e., $\mathbf{A}=(A_1,A_2)$,
UDP implies UDA for any pure state $\ket{\psi}$, regardless of the
dimension $d$.
\end{theorem}

To prove this theorem, recall that measuring $\mathbf{A}$ (i.e.,
measuring every observable in $\mathbf{A}$) for all quantum states
$\rho$ returns the set $C_m(\mathbf{A})$ given by Eq.~\eqref{eq:Cm}.
We know that $C_m(\mathbf{A})$ is a convex set, meaning for any
$\vec{x},\vec{y}\in C_m(\mathbf{A})$, we have
$(1-s)\vec{x}+s\vec{y}\in C_m(\mathbf{A})$ for any $0<s<1$.

For pure states, the corresponding set of average values is given by
$W_m(\mathbf{A})$ as defined in Eq.~\eqref{eq:Wm}. Unlike
$C_m(\mathbf{A})$, $W_m(\mathbf{A})$ in general is not convex.
Nevertheless, it is easy to see that $W_m(\mathbf{A})=C_m(\mathbf{A})$
when $W_m(\mathbf{A})$ is convex.

For $m=2$, the Hausdorff--Toeplitz theorem
\cite{Toeplitz:1918fv,Hausdorff:1919ip} gives convexity of the
numerical range of any operator, which in turn shows that
$W_2(\mathbf{A})$ is convex. We explain it briefly here. For any
operator $B$ acting on a Hilbert space $\mathcal{H}_d$, the numerical
range of $B$ is the set of all complex numbers
$\bra{\psi}B\ket{\psi}$, where $\ket{\psi}$ ranges over all pure
states in $\mathcal{H}_d$.

Note that one can always write $B$ as
\begin{eqnarray}
B&=&\frac{1}{2}[(B+B^{\dag})+(B-B^{\dag})]\nonumber\\
&=&\frac{1}{2}[(B+B^{\dag})+i(-iB+iB^{\dag})].
\end{eqnarray}
If we define $A_1:=(B+B^{\dag})/2$ and $A_2:=(-iB+iB^{\dag})/2$ then
clearly both $A_1$ and $A_2$ are Hermitian. Then $W_2(\mathbf{A})$ is
nothing but the numerical range of $B = A_1 + iA_2$ and hence is
convex.

Furthermore, by studying the properties of the numerical range, it was
shown in~\cite{Embry:1970up} (using different terminology) that if a
pure state $\ket{\psi}$ is UDP, the point
$\vec{x}:=\mathbf{A}(\ket{\psi})$ must be an extreme point of
$W_2(\mathbf{A})$. Here $\vec{x}$ is an extreme point of the convex
set $W_2(\mathbf{A})$ if there do not exist $\vec{y},\vec{z}\in
W_2(\mathbf{A})$, such that $\vec{x} = (1-s)\vec{y}+s\vec{z}$ for some
$0<s<1$.

Because $W_2(\mathbf{A})=C_2(\mathbf{A})$, $\vec{x}$ is also an
extreme point of $C_2(\mathbf{A})$. One can further show that for any
extreme point $\vec{x}$ of $C_2(\mathbf{A})$, and any quantum state
$\rho$ with $\mathbf{A}(\rho)=\vec{x}$, any pure quantum state
$\ket{\phi}$ in the range of $\rho$ will also have
$\mathbf{A}(\ket{\phi})=\vec{x}$. This then implies that if a pure
state $\ket{\psi}$ is UDP by measuring $\mathbf{A}$, it must also be
UDA, which proves the theorem.

Again, all the technical details of the proof will be presented in
Appendix C.

In an attempt to extend Theorem~\ref{th:Num} to the $m \geq 3$ case, a
natural question that one could ask is whether or not UDP implies UDA
whenever $W_m(\mathbf{A})$ is convex. Unfortunately this is not the
case, as demonstrated by the following example.

For the qutrit case ($d=3$), consider the observables
$\mathbf{A}=(M_1, M_2, M_3)$, where the $M_i$s are the Gell-Mann
matrices given in Eq.~\eqref{eq:GellMan}. These are the Pauli
operators embedded in the qutrit space. It is easily verified that in
this case, $W_m(\mathbf{A})$ is the Bloch sphere together with its
interior and is thus convex. Nonetheless, the unique pure state
compatible with measurement result $(0,0,0)$ is the state $\ket{2}$,
even though there are many mixed states sharing this measurement
result, such as $\frac{1}{2}(\oprod{0}{0} + \oprod{1}{1})$.

Therefore, although the Hausdorff--Toeplitz theorem
\cite{Toeplitz:1918fv,Hausdorff:1919ip} is famous for showing the
convexity of numerical range of any operator, there is indeed a deeper
reason than just the convexity of the numerical range which governs
the validity of Theorem~\ref{th:Num}. We leave the more detailed
discussion to Appendix~C.

\section{V. Symmetry of the state space}

In this section, we discuss some circumstances where UDP implies UDA
in a more general context where more than two observables are
measured, i.e., $m>2$. Our focus is on the symmetry of the set of all
quantum states. For a $d$-dimensional Hilbert space $\mathcal{H}_d$ we
denote this set of states by $K_d$, that is
\begin{equation}
K_d=\{\rho \mid \rho\ \text{acts}\ \text{on}\ \mathcal{H}_d,
\tr(\rho)=1 \}.
\end{equation}
Note that $K_d$ is convex, as we know that for any $\rho_1,\rho_2\in
K_d$, $(1-s)\rho_1+s\rho_2\in K_d$ for all $0<s<1$. Furthermore, the
extreme points of $K_d$ are all the pure states. $K_d$ is also called
the state space for all the operators acting on $\mathcal{H}_d$.

We now explain the intuition. If $K_d$ has a certain symmetry, then
two pure states $\ket{\psi_1}$ and $\ket{\psi_2}$ that are `connected'
by the symmetry will give the same measurement results, and states
$\ket{\psi}$ fixed by the symmetry will also be fixed by the
projection onto the space of observables. In this situation, UDP
implies UDA for all pure states.

To make this intuition concrete, let us first consider an example for
$d=2$, i.e., the qubit case. We know that $K_d$ can be parameterized
as in Eq.~\eqref{eq:rho}, where for $d= 2$,
$\lambda_1=X,\lambda_2=Y,\lambda_3=Z$ are chosen as Pauli matrices
given in Eq.~\eqref{eq:Pauli}. Here $K_d$ is the Bloch ball as shown
in FIG.~\ref{fig:example}. The Bloch ball is clearly a convex set and
the extreme points are those pure states on the boundary, which give
the Bloch sphere.
\begin{figure}[htb]
\includegraphics[scale=0.7]{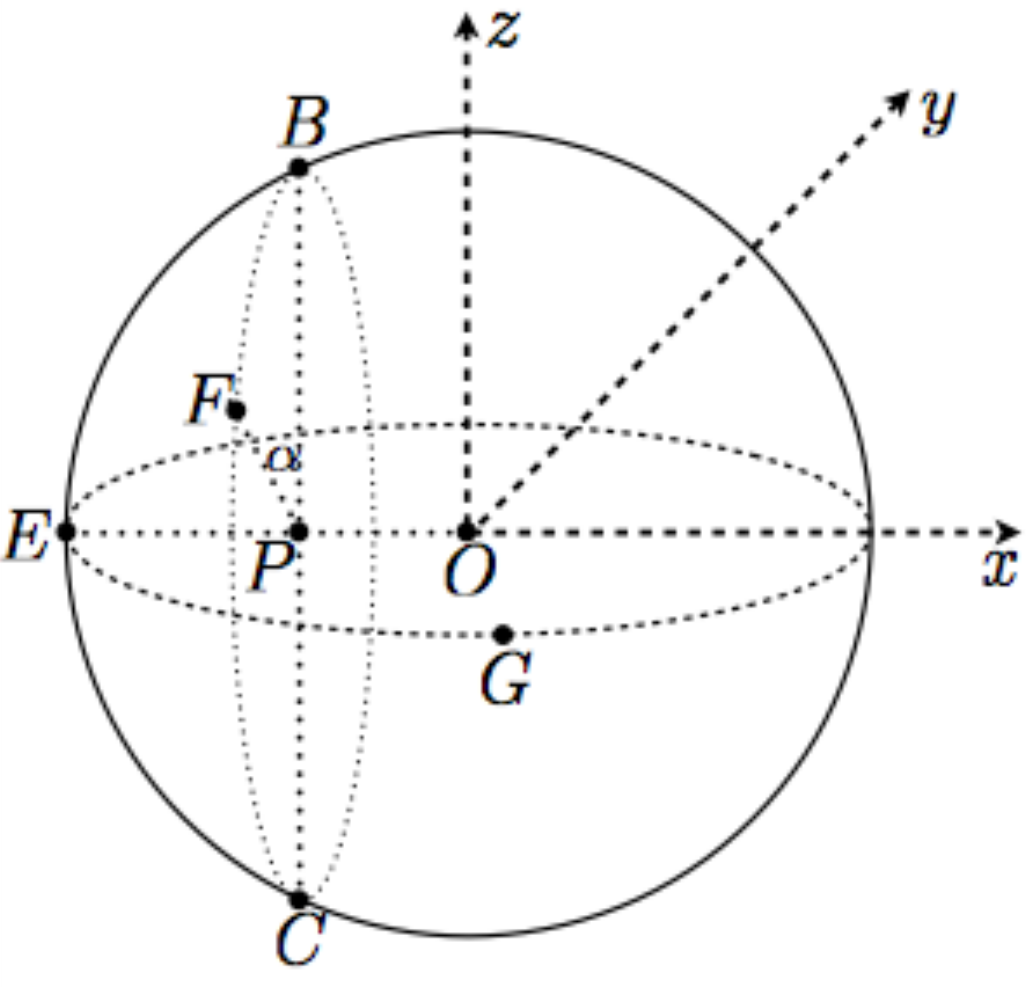}
\caption{Symmetry of the Bloch ball}
\label{fig:example}
\end{figure}

We know that geometrically, measuring the observables in $\mathbf{A}$
corresponds to the projection onto the plane spanned by $\mathbf{A}$.
For example, if we measure the Pauli $X$ and $Y$ operators, then
geometrically this corresponds to the projection of the Bloch ball
onto the $xy$ plane. Since the Bloch ball has reflection symmetry with
respect to the $xy$ plane, two pure states (e.g. points $B$ and $C$)
connected by that symmetry will project onto the same measurement
result $P$, as will all mixtures of $B$ and $C$. Hence neither UDP nor
UDA hold for such pure states for measuring $X$ and $Y$. On the other
hand, pure states fixed by the reflection symmetry are also fixed by
the projection onto the $xy$ plane. These are precisely the points on
the Bloch sphere that are in the $xy$ plane (e.g. the points $E$ and
$G$ in FIG.~\ref{fig:example}), and for such pure states both UDP and
UDA hold. Therefore, for the observables $X, Y$ we conclude that UDP =
UDA.

Now let us look at another case where we only measure the Pauli $X$
operator. Consider the group of symmetries of the Bloch ball
consisting of rotation around the $x$ axis. (Rotation by angle
$\alpha$, is shown in FIG.~\ref{fig:example}. In that figure, point
$B$ will become point $F$ after this particular rotation, and indeed
both points $B$ and $F$ yield the same measurement result, which is
represented by point $P$ on the $x$ axis.) Note that two points on the
Bloch sphere will project to the same measurement result on the $x$
axis if and only if they are in the same orbit under the rotation
group. Thus a measurement result will come from a single pure state
exactly when that pure state is a fixed point, and hence either both
or neither of UDP and UDA hold for each pure state. For example, the
point $E$ is fixed by the rotation, and is uniquely determined by the
measurement of $X$ among all states. $E$ corresponds to the $-1$
eigenstate of the Pauli $X$ operator. Therefore, the rotational
symmetry of the Bloch ball along the $x$ axis gives UDP =UDA for any
pure state when measuring the Pauli $X$ operator, which corresponds to
the $x$ axis.

Mathematically, a symmetry of $K_d$ is an affine automorphism of
$K_d$. If $U \in M_d$ is unitary, the map taking $\rho$ to $U\rho
U^\dagger$ is such an affine automorphism (which for $d = 2$ will just
be rotation around some axis of the Bloch ball). For instance, the
rotation symmetry along the $x$ axis by an angle $\alpha$ is given by
conjugation by the unitary operator $\exp(-iX\alpha/2)$. If $V$ is the
conjugate linear map given by complex conjugation in the computational
basis ($V\ket{\psi} = \ket{\psi^*}$), then the map taking $\rho$ to
$V\rho V^\dagger$ is the transpose map. For $d = 2$, this map is
reflection of the Bloch ball in the xy-plane.

Recall that for a set of observables $\mathbf{A}$, we denote the real
linear span by $\mathcal{S}(\mathbf{A})$. When discussing the
uniqueness problems, it makes no difference if we append the identity
operator to $\mathbf{A}$. Let us then assume $\mathbf{A}=(I,
A_1,\ldots, A_m)$. We are now ready to put our intuition into a
theorem.
\begin{theorem}
\label{th:symmetry}
Assume there exists a compact group $G$ of affine automorphisms of
$K_d$ whose fixed point set is $K_d\bigcap S(\mathbf{A})$. Then each
pure state acting on $\mathcal{H}_d$ which is UDP for measuring
$\mathbf{A}$ is also UDA.
\end{theorem}

In the first example above, the group for the reflection consists of
the two element group generated by the reflection. In the rotation
example, we can take the group to consist of all rotations around the
given axis. We will leave the detailed mathematical proof of Theorem
\ref{th:symmetry} to Appendix D, where operator algebras are one
ingredient of the proof.

To motivate some further consequences of Theorem~\ref{th:symmetry},
consider a simple example. If $\mathbf{A}$ consists of a basis of
diagonal matrices (i.e., a set of mutually commuting observables),
then for any pure state, UDP implies UDA by Theorem~\ref{th:symmetry}.
Here the group of symmetries can be taken to be conjugation by all
diagonal unitaries. This group has fixed point set $K_d\bigcap
S(\mathbf{A})$. In a more general case, if the complex span of $
S(\mathbf{A})$ is a *-subalgebra of the operators acting on
$\mathcal{H}_d$, then UDP = UDA for all pure states for measuring
$\mathbf{A}$. This is a natural corollary of Theorem~\ref{th:symmetry}
that we will also discuss in detail in Appendix D.

\section{VI. Conclusion and Discussion}

In this work, we have discussed the uniqueness of quantum states
compatible with given results for measuring a set of observables. For
a given pure state, we consider two different types of uniqueness, UDP
and UDA. We have taken the first step to study their relationship
systematically. In doing so we have established a number of results,
but also leave with many open questions.

First of all, although in general UDP and UDA are evidently different
concepts, their difference is surprisingly `not that large'.
Specifically in the sense of general counting of the number of
variables one needs to measure to uniquely determine all pure states
in a $d$ dimensional Hilbert space. Compared to full quantum
tomography which requires $d^2-1$ variables measured to uniquely
determine any quantum state, the $5d-7$ observables we have
constructed to uniquely determine any pure state among all states is a
significant improvement. It is indeed larger than the $4d-5$
observables given in \cite{HTW11} to uniquely determine any pure state
among all pure states, but the difference is only linear in $d$. We do
not know whether there could be another construction for which we
could further close the linear difference between UDA and UDP, to
leave only a constant gap for large $d$.

When the Hilbert space is a multipartite quantum system, and the
observables correspond to the RDMs, we focused on the situation when
`almost all pure states' are uniquely determined. We considered a
$3$-particle system with Hilbert space
$\mathcal{H}_d=\mathcal{H}_{d_1}\otimes\mathcal{H}_{d_2}\otimes\mathcal{H}_{d_3}$,
and showed that if $d_1\geq\min(d_2,d_3)$, then almost all pure states
are UDA by their $2$-RDMs of particle pairs $\{1, 2\}$ and $\{1, 3\}$.
This improves the results of \cite{LW02}, where $d_1\geq d_2+d_3-1$ is
required; however it still leaves a gap compared to the Diosi UDP
result which states that for $d_1\geq 2$, almost all pure states are
UDP by their $2$-RDMs of particle pairs $\{1, 2\}$ and $\{1, 3\}$.
Because our proof only gives a sufficient condition for UDA, we do not
know whether it can be further improved. We also do not have an
example showing there is indeed gap between UDA and UDP for almost all
three-particle pure states to be uniquely determined by $2$-RDMs of
particle pairs $\{1, 2\}$ and $\{1, 3\}$.

Finally, we considered situations for which we can show that UDP
implies UDA. These include: (i) the general $2$-qubit system; (ii) the
$3$-qubit system when we consider uniqueness for almost all pure
states and the measurements corresponds to $2$-RDMs; (iii) when only
two observables are measured; and (iv) the observables measured
correspond to some symmetry of the state space. However we do not know
how far we are from enumerating all the possible situations that UDP
implies UDA, when considering uniqueness for all pure states or almost
all pure states. In principle one can even consider the relationship
between UDP and UDA for special subsets of pure states.

We believe our systematic study of the uniqueness of quantum states
compatible with given measurement results shed light on several
aspects of quantum information theory and its connection to different
topics in mathematics. These include quantum tomography and the space
of Hermitian operators, unique ground states of local Hamiltonians and
general solutions to certain linear equations, measurements and
numerical ranges of operators, and the geometric meaning of
measurements and the symmetry of state space. We thus conclude with
several open questions that we believe warrant further investigation.

\section*{Appendix A: Proof of Theorem~\ref{th:UDA} and~\ref{th:UDAq}}

Theorem~\ref{th:UDA} can be implied by the following Lemma.

\begin{lemma}\label{lemma:2+-}
  There exists a set of $m = d^2-5d+6$ linearly independent Hermitian matrices
  $H_1, H_2, \ldots, H_m \in M_d(\C)$, such that the Hermitian matrix
  \begin{equation*}
    \sum_{j=1}^m r_j H_j
  \end{equation*}
  has at least two positive eigenvalues for any nonzero real vector
  $r=(r_j) \in \R^m$.
\end{lemma}

\begin{proof}
  We prove the statement by giving an explicit construction. Our proof
  is motivated by and similar to the diagonal filling technique used
  in Ref.~\cite{CMW08}, but along the other direction of the diagonals.

  We will need the lemma 9 from Ref.~\cite{CMW08} about totally
  non-singular matrix, which we restate as Lemma~\ref{lem:cmw08} in
  the following. For simplicity, we also assume that the totally
  non-singular matrix is real. Therefore, for any length $L \in \N$
  and $L \ge 2$, there is $L-1$ linearly independent real vectors such
  that every nonzero linear combination of them has at least $2$
  nonzero entries.

  Let $H = (H_{jj'}) \in M_d(\C)$ be a matrix. We will always fix the
  diagonal to be zero, namely $H_{jj}=0$ for $0\le j\le d-1$. In the
  upper triangular part of the matrix not including the diagonal,
  there are $2d-3$ lines of entries parallel to the antidiagonal. That
  is, each line contains entries $H_{jj'}$ with $j<j'$ and $j+j'=k$
  where $k$ goes from $1$ to $2d-3$. We will call it the $k$-th line
  of the matrix in the following. We also call the set of entries
  $H_{jj'}$ with $j+j'=k$ the $k$-th antidiagonal. It is easy to see
  that the length $L_k$ of the $k$-th line is
  \begin{equation*}
    L_k =
    \begin{cases}
       \displaystyle \Bigl[ \frac{k+1}{2} \Bigr] &
       \text{for } k\le d-1,\\[.5em]
       \displaystyle \Bigl[ \frac{2d-1-k}{2} \Bigr]
       & \text{otherwise.}
    \end{cases}
  \end{equation*}
  So the length $L_k \ge 2$ for $3\le k \le 2d-5$, and we can find
  $L_k-1$ real vectors for which every nonzero linear combination has
  at least $2$ nonzero entries. For each of the $L_k-1$ vectors, we
  can form two Hermitian matrices. One of them is the symmetric one
  whose $k$-th line is filled with the vector, and the lower
  triangular part determined by the Hermitian condition. Such a matrix
  is a real symmetric matrix having nonzero entries only on the $k$-th
  antidiagonal. We will call it a real $k$-th line matrix. The other
  is the one with $k$-th line filled with the vector multiplied by
  $i=\sqrt{-1}$, and lower part is determined by the Hermitian
  condition. This is a matrix consisting of purely imaginary entries
  on the $k$-th antidiagonal and we call it an imaginary $k$-th line
  matrix.

  Now we prove that the constructed matrices satisfy our requirement.
  First we prove that the matrices are linearly independent. It
  suffices to show that the matrices of nonzero $k$-th line is
  linearly independent. Let $\{v_j\}$ be the set of linearly
  independent real vectors chosen for the $k$-th line. We need to show
  that $\{ (v_j, v_j), (i v_j, -i v_j) \}$ is linearly independent
  over $\C$. If the contrary is true, that is, there exists complex
  numbers $c_j,d_j$ not all zero such that
  \begin{equation*}
    \sum_j c_j (v_j, v_j) + \sum_j d_j (i v_j, -i v_j) = 0.
  \end{equation*}
  This is equivalent to
  \begin{equation*}
    \begin{split}
          \sum_j c_j v_j + i \sum_j d_j v_j & = 0\\
          \sum_j c_j v_j - i \sum_j d_j v_j & = 0.
    \end{split}
  \end{equation*}
  From the above two equations, we get $\sum_j c_j v_j = 0$ and
  $\sum_j d_j v_j =0$ which is a contradiction.

  Next, we prove that for any nonzero real coefficient $r\in \R^m$,
  the matrix $H=\sum r_j H_j$ has at least two positive eigenvalues.
  Let $k_0$ be the largest $k$ such that there is a $k$-th line matrix
  $H_j$ whose coefficient $r_j$ is nonzero. Then, either the real
  $k_0$-th line matrices or the imaginary ones have nonzero
  coefficients. By the construction, this implies that there is at
  least two nonzero entries on the $k_0$th line of the matrix $H$. Let
  the nonzero entries be $a,b \in \C$. We then have a principle
  submatrix of $H$ that has the form
  \begin{equation*}
    \begin{pmatrix}
      0 & x & y & a\\
      \bar{x} & 0 & b & 0\\
      \bar{y} & \bar{b} & 0 & 0\\
      \bar{a} & 0 & 0 & 0
    \end{pmatrix},
  \end{equation*}
  where $x,y$ are two unknown number and $\bar{a}$ represents the
  complex conjugate of $a$. This matrix has trace $0$ and determinant
  $|ab|^2$. Therefore, it has exactly two positive eigenvalues. As it
  is a principle submatrix of matrix $H$, follows from
  Theorem~\ref{thm:horn90}, $H$ has at least two positive eigenvalues.

  The number of matrices thus constructed is the summation
  \begin{equation*}
    m=\sum_{k=3}^{2d-5} 2(L_k-1),
  \end{equation*}
  which can be computed to be
  \begin{equation*}
    d^2-5d+6.
  \end{equation*}
\end{proof}

Discussion: We note that our construction will also imply that the
matrix has at least two negative eigenvalues, thus at least rank $4$.
But our bound is even better than the $(d-3)^2$ bound on the dimension
of subspaces in which every matrix has rank $\ge 4$. This is not a
contradiction as we are considering all real combinations. For
example, the case of $d=4$ has two matrices for our purpose, namely
\begin{equation*}
  H_1 = \begin{pmatrix}
    0 & 0 & 0 & 1\\
    0 & 0 & 1 & 0\\
    0 & 1 & 0 & 0\\
    1 & 0 & 0 & 0
  \end{pmatrix} \text{ and }
  H_2 = \begin{pmatrix}
    0 & 0 & 0 & i\\
    0 & 0 & i & 0\\
    0 & -i & 0 & 0\\
    -i & 0 & 0 & 0
  \end{pmatrix}.
\end{equation*}
These two matrices do satisfy our requirements, but their span
over $\C$ contains a rank $2$ matrix $H_1+i H_2$.

Generalization: Similarly, length $L_k \ge q+1$ for $2q+1\le k \le
2d-2q-3$, and we can find $L_k-q$ real vectors for which every nonzero
linear combination has at least $q+1$ nonzero entries. For each of the
$L_k-q$ vectors, we can also form two Hermitian matrices. Such
constructed matrices are linearly independent and any real linear
combination has at least $q+1$ positive eigenvalues.

We restate our result as Lemma~\ref{lemma:rank-q} which will lead to
Theorem~\ref{th:UDA}.

\begin{lemma}\label{lemma:rank-q}
  There exists a set of $m = d^2-(4q+1)d+(4q^2+2q)$ linearly
  independent Hermitian matrices $H_1, H_2, \ldots, H_m \in M_d(\C)$,
  such that the Hermitian matrix
  \begin{equation*}
    \sum_{j=1}^m r_j H_j
  \end{equation*}
  has at least $q+1$ positive eigenvalues for any nonzero real vector
  $r=(r_j) \in \R^m$.
\end{lemma}

We just follow the lines of the proof of Lemma~\ref{lemma:2+-}. To
complete our argument, we need to show that any $2(q+1)$ by $2(q+1)$
invertible traceless, Hermitian, upper left triangler matrix has
exactly $q+1$ positive eigenvalues.

Let's prove this claim by induction. When $q=1$, it is already known.
Let's assume this claim holds true for any $q\leq r$. Then for
$q=r+1$, we can write such matrix $A$ in the following form
\begin{equation*}
  \begin{pmatrix}
    0 & x_1 & x_2 & \cdots & x_{2r+1} & x_{2r+2} &a\\
    \bar{x_1} &0 & y_1 &\cdots  &y_{2r}& b & 0\\
    \bar{x_2} & \bar{y_1} & 0 &\cdots& c & 0 & 0\\
    \vdots & \vdots & \vdots & \vdots & \vdots & \vdots\\
    \bar{a} & 0 & 0 & 0& 0 & 0 & 0
  \end{pmatrix}.
\end{equation*}

One may observe that, by deleting the first and the last rows/columns,
we will have a $2r$ by $2r$ invertible, traceless, Hermitian, upper
left triangler submatrix.

From our assumption, this submatrix has exactly $r$ positive
eigenvalues which means $A$ has at least $r$ positive eigenvalues and
at least $r$ negative eigenvalues.

Note that its determinant equals to $(-1)^{q+1} |ab\cdots|^2$. This
follows that $A$ has exactly $r+2$ positive eigenvalues which
completes our argument.

The number of matrices thus constructed is the summation
\begin{eqnarray*}
  m&=&2\sum_{k=2q+1}^{2d-2q-3}(L_k-q)\\
  &=&2\sum\limits_{k=2q+1}^{d-1}(\Bigl[ \frac{k+1}{2}
  \Bigr]-q)+2\sum\limits_{k=d}^{2d-2q-3}(\Bigl[ \frac{2d-1-k}{2}
  \Bigr]-q)\\
  &=& d^2-(4q+1)d+(4q^2+2q).
\end{eqnarray*}

\begin{lemma}[Lemma~9, \cite{CMW08}]
  \label{lem:cmw08}
  Let $M$ be a $d$ by $d$ totally non-singular matrix, with $d\ge
  n$. Let $v$ be any linear combination of $n$ of the columns of $M$.
  Then $v$ contains at most $n-1$ zero elements.
\end{lemma}

\begin{theorem}[Theorem~4.3.15, \cite{Horn90}]
\label{thm:horn90}
Let $A$ be a $n$ by $n$ Hermitian matrix, let $r$ be an integer with
$1\leq r \leq n$, and let $A_r$ denote any $r$ by $r$ principle
submatrix of $A$ (obtained by deleting $n-r$ rows and the
corresponding columns from $A$). For each integer $k$ such that $1\leq
k\leq r$ we have
\begin{eqnarray*}
\lambda_k^{\uparrow}(A)\leq \lambda_k^{\uparrow}(A_r)\leq \lambda_{k+n-r}^{\uparrow}(A).
\end{eqnarray*}
\end{theorem}

\section*{Appendix B: Proof of Theorem~\ref{th:RDM2}}
\begin{theorem}
  Almost every tripartite density operator $\rho \in
  \mathcal{B}(\mathcal{H}_{d_1}\otimes \mathcal{H}_{d_2}\otimes
  \mathcal{H}_{d_3})$ with rank no more than $\lfloor
  \frac{d_{1}}{d_{3}}\rfloor$ is UDA by its 2-RDMs of particle pairs
  $\{1,2\}$ and $\{1,3\}$.
\end{theorem}

\begin{proof}
  For any $\rho_{123} \in \mathcal{B}(\mathcal{H}_{d_1}\otimes
  \mathcal{H}_{d_2}\otimes \mathcal{H}_{d_3})$, we can choose
  $\ket{\phi}_{1234}$ to be the pure state whose 3-RDM of particles
  $\{1,2,3\}$ is exactly $\rho_{123}$. We can further assume
  $d_{4}\leq \rank \rho$.

Without loss of generality, we can assume
\begin{eqnarray*}
\ket{\phi}_{1234} = \sum\limits_{i_{1},i_{2},i_{3},i_{4}}
\lambda_{i_{1}i_{2}i_{3}i_{4}} \ket{i_{1}}_1\ket{i_{2}}_2 \ket{i_{3}}_3\ket{i_{4}}_4.
\end{eqnarray*}

If there is another $\sigma_{123}$ agrees with $\rho_{123}$ in
subsystems $\{1,2\}$ and $\{1,3\}$. Then we can find some pure state
$\ket{\psi}_{12345}$ whose 3-RDM of particle set $\{1,2,3\}$ is
$\sigma_{123}$.

In general, $\ket{v_{i_{3}i_{4}}}_{12} = \sum\limits_{i_{1}i_{2}}
\lambda_{i_{1}i_{2}i_{3}i_{4}} \ket{i_{1}}_1\ket{i_{2}}_2$ are
linearly independent and they will span the support of $\rho_{12}$.

Hence, any pure state $\ket{\psi}_{12345}$ which agrees with
$\ket{\phi}_{1234}$ in subsystem $\{1,2\}$ can be expanded as the
following.
\begin{eqnarray*}
&& \ket{ \psi}_{12345}\\
&=& \sum \limits_{i_{3},i_{4}} \ket{v_{i_{3}i_{4}}}_{12}
\ket{E_{i_{3}i_{4}}}_{345}\\
&=& \sum \limits_{i_{3},i_{4}} \sum \limits_{i_{1},i_{2}}
\lambda_{i_{1}i_{2}i_{3}i_{4}} \ket{i_{1}}_1 \ket{i_{2}}_2 \sum
\limits_{i_{3}^{ \prime},i_{4}^{ \prime}}  \ket{i_{3}^{ \prime}}_3
\ket{i_{4}^{ \prime}}_4 \ket{e_{i_{3}i_{4},i_{3}^{ \prime}i_{4}^{
      \prime}}}_5\\
&=& \sum \limits_{i_{1},i_{2},i_{3},i_{4},j_{3},j_{4}}
\lambda_{i_{1}i_{2}j_{3}j_{4}} \ket{i_{1}}_1 \ket{i_{2}}_2
\ket{j_{3}}_3 \ket{j_{4}}_4 \ket{e_{i_{3}i_{4},j_{3}j_{4}}}_5.
\end{eqnarray*}

Recall that $\ket{\psi}_{12345}$ and $\ket{\phi}_{1234}$ agree in
subsystem $\{1,3\}$, we will have
\begin{eqnarray*}
\tr_{\{2,4,5\}}(\ket{\psi}\bra{\psi}_{12345}) = \tr_{\{2,4\}}(
\ket{\phi} \bra{\phi}_{1234}).
\end{eqnarray*}

Firstly, let's look into the left hand side.
\begin{eqnarray*}
  &&\tr_{\{2,4,5\}}(\ket{\psi}\bra{\psi}_{12345})\\
  &=&\tr_{\{2, 4,
    5\}}(\sum\limits_{{i_{1},i_{2},i_{3},i_{4},j_{3},j_{4},}\atop{i_{1}^{\prime},i_{2}^{\prime},i_{3}^{\prime},i_{4}^{\prime},j_{3}^{\prime},j_{4}^{\prime}}}
  \lambda_{i_{1}i_{2}i_{3}i_{4}}\lambda_{i_{1}^{\prime}i_{2}^{\prime}i_{3}^{\prime}i_{4}^{\prime}}^{\ast}
  \cdot \ket{i_{1}}\bra{i_{1}^{\prime}}_1\\
  &&\otimes \ket{i_{2}}\bra{i_{2}^{\prime}}_2 \otimes
  \ket{j_{3}}\bra{j_{3}^{\prime}}_3\otimes\ket{j_{4}}\bra{j_{4}^{\prime}}_4\otimes
  \ket{e_{i_{3}i_{4},j_{3}j_{4}}}\bra{e_{i_{3}^{\prime}i_{4}^{\prime},j_{3}^{\prime}j_{4}^{\prime}}}_5)\\
  &=&\sum\limits_{{i_{1},i_{2},i_{3},i_{4},j_{3},}\atop{j_{4},i_{1}^{\prime},i_{3}^{\prime},i_{4}^{\prime},j_{3}^{\prime}}}\lambda_{i_{1}i_{2}i_{3}i_{4}}\lambda_{i_{1}^{\prime}i_{2}^{\prime}i_{3}^{\prime}i_{4}^{\prime}}^{\ast}\langle
  e_{i_{3}^{\prime}i_{4}^{\prime},j_{3}^{\prime}j_{4}}
  \ket{e_{i_{3}i_{4},j_{3}j_{4}}} \\
  &&\cdot\ket{i_{1}}\bra{i_{1}^{\prime}}_1\otimes\ket{j_{3}}\bra{j_{3}^{\prime}}_3\\
  &=&\sum\limits_{i_{1},j_{3},i_{1}^{\prime},j_{3}^{\prime}} (\sum\limits_{i_{2},i_{3},i_{4},j_{4},i_{3}^{\prime},i_{4}^{\prime}}\lambda_{i_{1}i_{2}i_{3}i_{4}} \lambda_{i_{1}^{\prime}i_{2}^{\prime}i_{3}^{\prime}i_{4}^{\prime}}^{\ast}\cdot \\
  &&\langle e_{i_{3}^{\prime}i_{4}^{\prime},j_{3}^{\prime}j_{4}} \ket{e_{i_{3}i_{4},j_{3}j_{4}}}) \ket{i_{1}}\bra{i_{1}^{\prime}}_1\otimes\ket{j_{3}}\bra{j_{3}^{\prime}}_3.
\end{eqnarray*}

Then, for the right hand side,
\begin{eqnarray*}
  &&\tr_{\{2,4\}}(\ket{\phi}\bra{\phi}_{1,2,3,4})\\
  &=& \sum\limits_{i_{1},i_{2},i_{3},i_{4},i_{1}^{\prime}, i_{3}^{\prime}} \lambda_{i_{1}i_{2}i_{3}i_{4}}\lambda_{i_{1}^{\prime}i_{2}i_{3}^{\prime}i_{4}}^{\ast} \ket{i_{1}}\bra{i_{1}^{\prime}}_1\otimes \ket{i_{3}}\bra{i_{3}^{\prime}}_3\\
  &=& \sum\limits_{i_{1},j_{3},i_{1}^{\prime}, j_{3}^{\prime}} \sum\limits_{i_{2},j_{4}}\lambda_{i_{1}i_{2}j_{3}j_{4}}\lambda_{i_{1}^{\prime}i_{2}j_{3}^{\prime}j_{4}}^{\ast} \ket{i_{1}}\bra{i_{1}^{\prime}}_1\otimes \ket{j_{3}}\bra{j_{3}^{\prime}}_3.
\end{eqnarray*}

Combining the above two equations, we have
\begin{eqnarray*}
&&\sum\limits_{{i_{2},i_{3},i_{4},j_{4},}\atop{i_{3}^{\prime},i_{4}^{\prime}}}\lambda_{i_{1}i_{2}i_{3}i_{4}} \lambda_{i_{1}^{\prime}i_{2}^{\prime}i_{3}^{\prime}i_{4}^{\prime}}^{\ast}\langle e_{i_{3}^{\prime}i_{4}^{\prime},j_{3}^{\prime}j_{4}} \ket{e_{i_{3}i_{4},j_{3}j_{4}}}\\
&= &\sum\limits_{i_{2},j_{4}}\lambda_{i_{1}i_{2}j_{3}j_{4}}\lambda_{i_{1}^{\prime}i_{2}j_{3}^{\prime}j_{4}}^{\ast}
\end{eqnarray*}
for any $i_{1},i_{1}^{\prime}, j_{3},j_{3}^{\prime}$.

Similarly, follows from the fact that $\ket{\psi}_{12345}$ and
$\ket{\phi}_{1234}$ agree in subsystem $\{1,2\}$, we have the
following equation.
\begin{eqnarray*}
&&\sum\limits_{{i_{3},i_{4},j_{3},}\atop{j_{4},i_{3}^{\prime},i_{4}^{\prime}}}\lambda_{i_{1}i_{2}i_{3}i_{4}} \lambda_{i_{1}^{\prime}i_{2}^{\prime}i_{3}^{\prime}i_{4}^{\prime}}^{\ast}\langle e_{i_{3}^{\prime}i_{4}^{\prime},j_{3}j_{4}} \ket{e_{i_{3}i_{4},j_{3}j_{4}}}\\
&=&\sum\limits_{i_{3},i_{4}}\lambda_{i_{1}i_{2}i_{3}i_{4}}\lambda_{i_{1}^{\prime}i_{2}^{\prime}i_{3}i_{4}}^{\ast}
\end{eqnarray*}
for any $i_{1},i_{2},i_{1}^{\prime}, i_{2}^{\prime}$.

Let's denote $x_{p_3,p_4,q_3,q_4,p_{3}^{\prime}, p_{4}^{\prime},q_{3}^{\prime}}=\langle e_{p_{3}p_{4},q_{3}q_{4}}\ket{e_{p_{3}^{\prime}p_{4}^{\prime},q_{3}^{\prime}q_{4}}}$ for any $p_3,p_4,q_3,q_4,p_{3}^{\prime}, p_{4}^{\prime},q_{3}^{\prime}$.

If there is only one solution $\{x_{p_3,p_4,q_3,q_4,p_{3}^{\prime},
  p_{4}^{\prime},q_{3}^{\prime}}\}$ satisfies the above two linear
systems, then
\begin{eqnarray*}
&&\sigma_{123}\\
&=&\tr_{\{4,5\}}\ket{\psi}\bra{\psi}_{12345}\\
&=&\tr_{\{4,5\}}(\sum\limits_{{i_{1},i_{2},i_{3},i_{4},j_{3},j_{4},}\atop{i_{1}^{\prime},i_{2}^{\prime},i_{3}^{\prime},i_{4}^{\prime},j_{3}^{\prime},j_{4}^{\prime}}} \lambda_{i_{1}i_{2}i_{3}i_{4}}\lambda_{i_{1}^{\prime}i_{2}^{\prime}i_{3}^{\prime}i_{4}^{\prime}}^{\ast} \cdot\ket{i_{1}}\bra{i_{1}^{\prime}}_1\\
&&\otimes \ket{i_{2}}\bra{i_{2}^{\prime}}_2\otimes \ket{j_{3}}\bra{j_{3}^{\prime}}_3\otimes \ket{j_{4}}\bra{j_{4}^{\prime}}_4\otimes \ket{e_{i_{3}i_{4},j_{3}j_{4}}}\bra{e_{i_{3}^{\prime}i_{4}^{\prime},j_{3}^{\prime}j_{4}^{\prime}}})\\
&=&\sum\limits_{{i_{1},i_{2},i_{3},i_{4},j_{3},j_{4},}\atop{i_{1}^{\prime},i_{2}^{\prime},i_{3}^{\prime},i_{4}^{\prime},j_{3}^{\prime}}} \lambda_{i_{1}i_{2}i_{3}i_{4}}\lambda_{i_{1}^{\prime}i_{2}^{\prime}i_{3}^{\prime}i_{4}^{\prime}}^{\ast} \langle e_{i_{3}^{\prime}i_{4}^{\prime},j_{3}^{\prime}j_{4}}  \ket{e_{i_{3}i_{4},j_{3}j_{4}}}\\
&&\cdot \ket{i_{1}}\bra{i_{1}^{\prime}}_1\otimes  \ket{i_{2}}\bra{i_{2}^{\prime}}_2\otimes \ket{j_{3}}\bra{j_{3}^{\prime}}_3
\end{eqnarray*}
is completely determined.

There are generically $d_1^2 d_3^2+d_1^2d_2^2$ linearly independent
linear equations and $d_3^4d_4^2$ variables. $d_3^4d_4^2\leq d_3^4
(\rank{\rho})^2\leq d_3^4(\frac{d_1}{d_3})^2<d_1^2d_3^2+d_1^2d_2^2$
implies that there is at most one solution. Thus a generic low-rank
density operator $\rho_{123}$ is UDA by its 2-RDMs of particle pairs
$\{1,2\}$ and $\{1,3\}$.
\end{proof}

\section*{Appendix C: Proof of Theorem~\ref{th:Num}}

We begin by presenting without proof a result that was proved as
Theorem~1(i) in~\cite{Embry:1970up}. Before we state the result,
recall from Sec.~IV that $W_2(A)$ is the numerical range of $A_1 +
iA_2$ and is thus convex by the Hausdorff--Toeplitz theorem
\cite{Toeplitz:1918fv,Hausdorff:1919ip}. It therefore makes sense to
talk about extreme points of $W_2(A)$ in this case.
\begin{proposition}\label{prop:embry}
	Let $\vec{x} \in W_2(A)$. Then $\vec{x}$ is an extreme point of $W_2(A)$ if and only if
	\begin{align*}
		M_x := \{ \lambda\ket{\psi} : \lambda \in \mathbb{C}, {\mathbf A}(\ket{\psi}) = \vec{x}\}
	\end{align*}
	is a linear subspace of $\mathcal{H}_d$.
\end{proposition}

We are now in a position to prove Theorem~\ref{th:Num}.
\begin{proof}[Proof of Theorem~\ref{th:Num}]
  Suppose that $\ket{\psi}$ is UDP and define $\vec{x} :=
  (\bra{\psi}A_1\ket{\psi},\bra{\psi}A_2\ket{\psi})$. Then $M_x =
  \{\lambda\ket{\psi} : \lambda \in \mathbb{C}\}$, which is linear, so
  $\vec{x}$ is an extreme point of $W_2(A)$ by
  Proposition~\ref{prop:embry}. Because $W_2(A)=C_2(A)$ in this case
  by convexity, $\vec{x}$ is also an extreme point of $C_2(A)$.
  Suppose now that there exists a mixed state $\rho = \sum_i p_i
  \oprod{\psi_i}{\psi_i}$ with $\vec{x} = (\Tr(A_1 \rho), \Tr(A_2
  \rho))$. Then $\vec{x} = \sum_i p_i (\bra{\psi_i}A_1\ket{\psi_i},
  \bra{\psi_i}A_2\ket{\psi_i})$. Since $\vec{x}$ is extreme in
  $C_2(A)$, it follows that $\vec{x} = (\bra{\psi_i}A_1\ket{\psi_i},
  \bra{\psi_i}A_2\ket{\psi_i})$ for all $i$, which contradicts the
  fact that $\ket{\psi}$ is UDP unless each $\ket{\psi_i}$ is the same
  up to global phase (i.e., $\rho = \oprod{\psi}{\psi}$).
\end{proof}

Based on the proof of Theorem~\ref{th:Num}, we might expect that UDP
implies UDA for all pure states whenever $W_m({\mathbf A})$ is convex.
However, the example provided in Sec.~IV showed this not to be the
case. We now expand upon the reason for this apparent discrepancy,
which lies buried in the proof of Proposition~\ref{prop:embry}.

In the case when $W_m({\mathbf A})$ is convex, the ``only if''
implication of Proposition~\ref{prop:embry} still holds for arbitrary
$m$. However, the proof of the ``if'' implication relies on the fact
that if $\vec{x} := \mathbf{A}(\ket{\phi})$ and $\vec{y} :=
\mathbf{A}(\ket{\psi})$, then for any $s \in (0,1)$ we can find
$\alpha, \beta \in \mathbb{C}$ such that $s\vec{x} + (1-s)\vec{y} =
\mathbf{A}(\alpha\ket{\phi} + \beta\ket{\psi})$. In other words, the
proof of the proposition uses the fact that $W_m({\mathbf A})$ is not
only convex, but that convex combinations are in a sense well-behaved
between the input and output of $\mathbf{A}(\cdot)$. For convenience,
we refer to this property as \emph{strong convexity} for the remainder
of this section.

The standard proofs of the Hausdorff--Toeplitz theorem show that
strong convexity, not just convexity itself, always holds when $m =
2$. To see how strong convexity can fail when $m > 2$ even when
convexity holds, we again return to the example of Sec.~IV. In this
case, we have $\mathbf{A}(\ket{0}) = (0,0,1)$ and $\mathbf{A}(\ket{1})
= (0,0,-1)$. However, even though $W_3(A)$ is convex and thus there
exists a pure state $\ket{\psi}$ with $\mathbf{A}(\ket{\psi}) =
(0,0,0)$, the only such pure state is $\ket{\psi} := \ket{2}$, which
is not contained in the span of $\ket{0}$ and $\ket{1}$.

We might hope that strong convexity, rather than convexity itself,
provides the natural generalization of Theorem~\ref{th:Num}. That is,
we might hope that if $W_m({\mathbf A})$ is strongly convex, then UDP
implies UDA for all pure states. It turns out that this is a true but
vacuous statement --- if $W_m({\mathbf A})$ is strongly convex then it
must be the case that $m \leq 2$, so Theorem~\ref{th:Num} itself
applies directly. This fact seems to be implicit in many papers on the
joint numerical range, but we prove it here for completeness.

Before stating the result, we briefly note that we can assume without
loss of generality that ${\mathbf A}$ contains $I$ and is linearly
independent, as adding the identity to ${\mathbf A}$ has no effect on
convexity, UDA, or UDP, and furthermore these properties only depend
on the span of the observables in ${\mathbf A}$.
\begin{proposition}\label{prop:strong_convex}
  Let ${\mathbf A} = (I, A_1, \ldots, A_m)$ be a linearly independent
  set. Then $W_{m+1}({\mathbf A})$ is strongly convex if and only if
  $m \leq 2$.
\end{proposition}
\begin{proof}
  The ``if'' direction, as already mentioned, follows from any of the
  usual proofs of the Hausdorff--Toeplitz theorem.
	
  For the ``only if'' direction, suppose that that $W_{m+1}({\mathbf
    A})$ is strongly convex and assume (in order to get a
  contradiction) that $m \geq 3$. By \cite[Theorem~4.1]{LP99}, there
  exists $X \in M_{d,2}$ with $X^*X = I$ such that $X^* {\mathbf A} X
  := \{I, X^*A_1 X, \ldots, X^*A_m X\}$ spans all of $M_2$. By letting
  $\ket{\phi}$ and $\ket{\psi}$ be the column vectors of $X$, we see
  that strong convexity of $W_{m+1}({\mathbf A})$ immediately implies
  convexity of $W_{m+1}(X^* {\mathbf A} X)$. Since convexity of
  $W_{m+1}(X^* {\mathbf A} X)$ depends only on the span of $X^*
  {\mathbf A} X$, it follows that $W_{4}({\mathbf B})$ is also convex,
  where $\mathbf{B} := \{\lambda_0, \lambda_1, \lambda_2, \lambda_3\}$
  is the Pauli basis given by Equation~\eqref{eq:Pauli}. However, it
  is easily verified that $W_{4}({\mathbf B})$ is the Bloch sphere
  embedded in $4$-dimensional space and hence is not convex, which
  gives the desired contradiction.
\end{proof}

It thus seems that numerical range and convexity arguments are not
able to tell us anything non-trivial about the UDP/UDA problem beyond
the $m = 2$ case of Theorem~\ref{th:Num}.

\section*{Appendix D: Symmetries and UDP/UDA: Proof of Theorem~\ref{th:symmetry}}

We will see that if there is a compact group of symmetries whose fixed
point set is the linear span of observables $(A_1, \ldots, A_m)$, then
UDP for these observables implies UDA for any pure state. We start
with the following result about fixed points of compact groups.

\begin{theorem} \label{Haar}Let $G$ be a compact group of unitaries on
  a real or complex finite dimensional Hilbert space $H$, and let $L$
  be the set of fixed points of $G$. Let $\mu$ be Haar measure on $G$,
  and define $P:H \to H$ to be the linear map satisfying
\begin{equation} \label{Pdef} \langle P\xi, \eta \rangle   = \int_G \langle g\xi, \eta\rangle  \, d\mu(g)
\end{equation}
for $\xi, \eta \in H$.  Then $P$ is the orthogonal projection onto $L$,
  $Pg = gP = P$ for all $g \in G$, and $P$ is in the convex hull of $G$.
\end{theorem}

\begin{proof} Left and right invariance of Haar measure imply that $Pg
  = gP = P$ for all $g \in G$. The definition of $P$ implies that $L
  \subset \im P$. Now $gP = P$ for all $g \in G$ implies $\im P
  \subset L$, and hence $\im P = L$. Next, $\im P = L$ and the
  definition of $P$ give $P^2 = P$. To show that $P^\dagger = P$ we
  use the fact that the integrals of $f(g)$ and $f(g^{-1})$ are the
  same for Haar measure, together with the assumption that $G$ is a
  group of unitaries:
\begin{align}
  \langle \xi, P\eta\rangle &={ \langle P\eta, \xi\rangle^* } =\int_G
  {\langle g\eta, \xi\rangle^* } d\mu(g)\cr &= \int_G \langle \xi,
  g\eta\rangle d\mu(g) = \int_G \langle g^\dagger\xi, \eta\rangle
  d\mu(g)\cr &= \int_G \langle g^{-1}\xi, \eta\rangle d\mu(g) = \int_G
  \langle g\xi, \eta\rangle d\mu(g) \cr &= \langle P\xi, \eta\rangle
\end{align}

Finally, by the Alaoglu--Birkhoff mean ergodic theorem \cite[Prop.
4.3.4]{BR} $P$ is in the strong closure of the convex hull of $G$.
Since $H$ is finite dimensional, then the space of linear operators on
$H$ is also finite dimensional, so the convex hull of the compact set
$G$ is compact and hence closed.

\end{proof}

Symmetries of $K_d$ are given by conjugation by unitaries or by the
transpose map or by composition of these two types of symmetries. (An
affine automorphism of $K_d$ preserves transition probabilities, cf.
\cite{Shultz}, so this is a consequence of Wigner's theorem
\cite[233-236]{Wigner}.)

If we view the space of observables in $M_d$ as a real Hilbert space
(with the usual inner product $\langle X, Y\rangle = \tr(X Y)$), then
conjugation by unitaries and the transpose map both preserve this
inner product, so are given by unitaries on this Hilbert space.

If $L$ is a (real) linear subspace of observables containing the
identity, then $L$ will be the real linear span of $L \cap K_d$. Thus
any symmetry of $K_d$ will fix $L \cap K_d$ if and only if that
symmetry when extended to a map on $M_d$ fixes $L$. If $G$ is a
compact group of symmetries whose fixed point set is $L \cap K_d$,
then the corresponding maps on $M_d$ will have fixed point set $L$.

\begin{theorem} \label{symmetry} Let $A$ be a finite set of
  observables on $H_d$ with real linear span $L$. Assume there exists
  a compact group $G$ of affine automorphisms of $K_d$ whose fixed
  point set is $L \cap K_d$. Then each pure state which is UDP for
  measuring $A$ is also UDA.
\end{theorem}

\begin{proof} As discussed above, we may view $G$ as a compact group
  of unitaries with fixed point set $L$. Define $P$ as in Theorem
  \ref{Haar}. Fix a pure state $\rho$.

  Suppose first that $\rho \notin L\cap K_d$. Then there is some $g
  \in G$ such that $g(\rho) \not= \rho$. Since $Pg = P$, both
  $g(\rho)$ and $\rho$ are pure states with the same image in $L\cap
  K_d$ under the map $P$. Thus UDP fails for $\rho$ (and hence
  trivially UDA fails).

  Now suppose $\rho \in L\cap K_d$. Let $\sigma \in K_d$ be a pre
  image of $\rho$ under $P$. Then
  $$1 = \langle  P \sigma, \rho\rangle  = |\int_G \langle g\sigma,
  \rho\rangle \,d\mu(g)| \le \int_G ||g\sigma \|\|\rho\| \,d\mu(g)\le
  1,$$ and equality can hold only if $g\sigma = \rho$ for all $g$,
  i.e., if and only if $\sigma \in L$. Then $\sigma = P\sigma = \rho$,
  so for such $\rho$ both UDP and UDA hold.
\end{proof}

\begin{corollary}\label{cor:2dim} For $d=2$, for all pure states and
  all sets ${\bf A}$ of observables, UDP implies UDA.
\end{corollary}

\begin{proof}
  Let $A_1 = I, A_2, \ldots, A_m$ be observables in $M_2$ and let $L =
  S(A)$ be their real linear span. We will show that there is a finite
  group of affine automorphisms $G$ of the state space $K_d$ of $M_2$
  with fixed point set $L\cap K_d$. There are three cases, depending
  on the dimension of the fixed point set. The fixed point set in the
  Bloch sphere will be the central point, a diameter of the Bloch
  sphere, or the intersection of a plane (through the center) with the
  Bloch sphere. In each case reflection of the Bloch sphere in the
  fixed point set generates an order 2 group of affine automorphisms
  with fixed point set $L\cap K_d$. Now the corollary follows from
  Theorem \ref{symmetry}.
\end{proof}

\begin{corollary}\label{algebra} Let ${\bf A} = A_1,\ldots, A_p$ be
  observables in $M_d$. If the (complex) linear span of $\bf A$ is a
  *-subalgebra $\A$ of $M_d$, then UDP = UDA for pure states measured
  by these observables.
\end{corollary}

\begin{proof} UDP and UDA for a set of observables aren't affected if
  we include the identity among those observables, so hereafter we
  assume that $I_d \in {\bf A}$. Note that $\A$ is the linear span of
  the unitaries in $\A$. Furthermore, $\A$ is a von Neumann algebra
  containing the identity $I_d$, so by the bicommutant theorem
  \cite[Thm. 2.77]{AlfsenShultz} $(\A')' = \A$, where for $X \subset
  M_d$, $X'$ denotes the algebra of matrices that commute with all
  matrices in $X$. Combining these two statements shows that $\A$ is
  the set of matrices that commute with all unitaries in $\A'$, and
  thus is the set of fixed points of $G = \{{\rm Ad}_U \mid \text{U is
    a unitary in $\A'$} \}$. It follows that $L = \A_{sa}$ (the
  Hermitian matrices in $\A$) is the set of observables fixed by the
  compact group $G$. The corollary follows from Theorem
  \ref{symmetry}.

\end{proof}

\begin{example}
  Let ${\bf A} = \{E_{11}, \ldots, E_{dd}\}$. Then the complex linear
  span of ${\bf A}$ consists of the diagonal matrices. From Corollary
  \ref{algebra} it follows that for each pure state on $M_d$, UDP for
  ${\bf A}$ implies UDA.
\end{example}

We can generalize the last example by taking the *-algebra consisting
of diagonal observables with the restriction that certain diagonal
entries coincide. For example, if $d = 7$ we can look at diagonal
matrices of the form $\diag(a, a, b, b, b, c, d)$ whose linear span
will be 4 dimensional. The space of Hermitian members of this algebra
is four dimensional. If we choose 4 observables including the identity
spanning this space, then UDA = UDP for all pure states when measuring
these observables. (We could drop the identity from this list if we
wish.) In this way for any $d$ we can find a set of $k$ observables
for any $k \le d$ for which UDA = UDP.

We can also find many larger sets of observables for which UDA = UDP.
For example, for any $d$ we can consider the *-algebra of all block
diagonal matrices with $k$ blocks that are of size $d_i \times d_i$
for $1 \le k \le p$, where $d_1 + d_2 +\cdots d_p = d$. The subspace
of Hermitian matrices in this algebra has dimension $\sum_i d_i^2$, so
any such dimension is realizable as the number of observables in a set
of observables for which UDA = UDP holds for all pure states.

\section{Acknowledgements}

JC is supported by NSERC and NSF of China (Grant No. 61179030). ZJ
acknowledges support from NSERC, ARO. NJ is supported by the
University of Guelph Brock Scholarship and an NSERC Postdoctoral
Fellowship. DWK is supported by NSERC Discovery Grant 400160 and NSERC
Discovery Accelerator Supplement 400233. BZ is supported by NSERC and
CIFAR.

\end{document}